\documentclass[11pt, reqno]{amsart}
\usepackage{mathrsfs}
\usepackage{graphics}
\usepackage{latexsym,amssymb,amsmath,amscd,amscd,amsthm,amsxtra}
\usepackage[dvips]{graphicx}
\usepackage[dvips]{graphics,color}
\usepackage{amsfonts}
\usepackage[colorlinks,linkcolor=red]{hyperref}
   \def\lineno.sty{\texttt{\itshape lineno.sty}}
\def\Box{\vcenter{\vbox{\hrule\hbox{\vrule
     \vbox to 8.8pt{\hbox to 10pt{}\vfill}\vrule}\hrule}}}

\newtheorem{thm}{Theorem}[section]
\newtheorem{lem}[thm]{Lemma}
\newtheorem{cor}[thm]{Corollary}

\newtheorem{defn}[thm]{Definition}
\newtheorem{example}[thm]{Example}
\newtheorem{remark}{Remark}

\numberwithin{equation}{section}

\begin{document}
 \thispagestyle{empty}
\title{Perfect State Transfer on Abelian Cayley Graphs}

\author[Y. Tan, K. Feng, X.Cao]{Yingying Tan, Keqin Feng*, Xiwang Cao}

\address{Keqin Feng is with Department of Mathematical Sciences, Tsinghua
University, Beijing 100084, China, email: {\tt
kfeng@math.tsinghua.edu.cn}}
\address{Yingying Tan is with Department of Mathematics, Anhui University, Hefei 230601, Anhui Province, China. email: {\tt tansusan1@aiai.edu.cn}}
\thanks{Keqin Feng has supported by National Natural Science Foundation of China under grant numbers 11471178, 11571107 and the Tsinghua National Information Science and Technology Lab., Xiwang Cao has supported by National Natural Science Foundation of China under grant numbers 11771007, 61572027}

\address{Xiwang Cao is with the School of Mathematical Sciences, Nanjing University of Aeronautics and Astronautics, Nanjing 210016, China, and he is also with State Key Laboratory of Information Security, Institute of Information Engineering, Chinese Academy of Sciences, Beijing 100093, Chinaemail: {\tt xwcao@nuaa.edu.cn}}

\begin{abstract}
Perfect state transfer (PST) has great significance due to its applications in quantum information processing and quantum computation. In this paper we present a characterization on connected simple Cayley graph $\Gamma={\rm Cay}(G,S)$ having PST. We show that many previous results on periodicity and existence of PST of circulant graphs (where the underlying group $G$ is cyclic) and cubelike graphs ($G=(\mathbb{F}_2^n,+)$) can be derived or generalized to arbitrary abelian case in unified and more simple ways from our characterization. We also get several new results including answers on some problems raised before.
\end{abstract}

\keywords{perfect state transfer, Cayley graph, circulant graph, cubelike graph, integral graph, periodic graph}

\maketitle

\section{Introduction }

Let $\Gamma=(V,E)$ be a connected simple graph where $V$ and $E$ are the set of vertices and edges of $\Gamma$ respectively, $n=|V|\geq 2$, $A=A(\Gamma)=(a_{u,v})_{u,v\in V}$ be the adjacency matrix of $\Gamma$ defined by
\begin{equation*}
  a_{u,v}=\left\{\begin{array}{cc}
                   1 & \mbox{ if $(u,v)\in E$} \\
                   0 & \mbox{ otherwise}.
                 \end{array}
  \right.
\end{equation*}
$A$ is a $n\times n$ symmetric $\{0,1\}$-matrix and all eigenvalues of $A$ are real numbers. The transfer matrix of $\Gamma$ is defined by the following $n\times n$ unitary matrix
\begin{equation*}
  H(t)=H_{\Gamma}(t)=\exp(itA)=\sum_{s=0}^{+\infty}\frac{(itA)^s}{s!}=(H_{u,v}(t))_{u,v\in V}\ \ (t\in \mathbb{R},n=|G|, i=\sqrt{-1})
\end{equation*}

\begin{defn}For $u,v\in V$, $\Gamma$ is called to have perfect state transfer (PST) from $u$ to $v$ at time $t>0$ if $|H_{u,v}(t)|=1$. $\Gamma$ is called periodic at $u$ with periodic $t>0$ if $|H_{u,u}(t)|=1$.\end{defn}

It is known that if $\Gamma$ has PST from $u$ to $v$ at time $t$, $\Gamma$ also has PST from $v$ to $u$ at time $t$. Thus we call this as $\Gamma$ has PST between $u$ and $v$ at time $t$.

PST has great significance due to its applications in quantum information processing and quantum computation \cite{1, 8, 11, 12, 13, 19, 21, 22, 28}. The basic problem on this topic is to find graphs having PST. More exactly speaking, for a given connected simple graph $\Gamma=(V,E)$, and $u,v\in V$, define
\begin{eqnarray}\label{f-1}
  \widetilde{T}(u,v)&=&\{t\in \mathbb{R}: |H_{u,v}(t)|=1\}\\
 \nonumber T(u,v)&=&\widetilde{T}(u,v) \cap \mathbb{R}_{> 0}\\
 \nonumber &=&\{t>0: \Gamma \mbox{ has PST between $u$ and $v$ at time $t$}\}.
\end{eqnarray}
 We want to determine $T(u,v)$ for any pair $\{u,v\}$ of vertices, or at least, to determine if $\Gamma$ has PST ($T(u,v)$ is not empty set for some pair $\{u,v\}$).

 In 2004-2005, M. Christandl et al (\cite{12, 13}) showed that the path with two and three vertices and their Cartesian powers has PST between any two vertices with the longest distance. In his three papers (\cite{15, 16, 17}), C. Godsil surveys the progress up to 2011 on PST and periodicity of several families of graphs and explain the close relationship
 between this topic and algebraic combinatorics (spectrum of adjacency matrix, association scheme also see \cite{14}). In past decade, the research on this topic becomes very active \cite{2, 3, 4, 5, 6, 7, 10, 14, 18, 23, 24, 25, 27, 29}. Many results on periodicity and existence of PST have been obtained for distance-regular graphs \cite{14}, complete bipartite
   graphs \cite{27}, Hadamard diagonalizable graphs \cite{18}, variants of cycles \cite{3} and Johnson schemes \cite{2}. Particularly, the Cayley graphs over finite abelian group $G$ are paid much attention \cite{4, 5, 6, 7, 10, 23, 24, 25, 29}. Among many results, M. Ba${\rm \check{s}}$i${\rm \acute{c}}$ \cite{4} and W. Cheung-C. Godsil \cite{10} present a characterization of circulant graphs (the underlying group $G$ is cyclic) and cubelike graphs ($G=(\mathbb{F}_2^n,+)$) having PST.
 The main aim of this paper is to present a characterization of connected simple cayley graphs over arbitrary finite abelian group, called abelian Cayley graphs, having PST. More exactly speaking, for any abelian connected simple Cayley graph $\Gamma$, we determine the set $T(u,v)$ for arbitrary pair $(u,v)$ of vertices in $\Gamma$ (Theorem \ref{thm-main1} and \ref{thm-main}). From this characterization, many previous results can be derived or generalized to arbitrary abelian case in unified and more simple way. We also get several new results on periodicity and existence of PST on abelian Cayley graphs including to answer some problems raised in \cite{6, 15, 16}.

 In next section, we determine $T(u,v)$ for any pair of vertices $\{u,v\}$ in any abelian connected simple Cayley graph $\Gamma={\rm Cay}(G,S)$ (Theorem \ref{thm-main1} and \ref{thm-main}). It is known that such graph having PST should be integral. Namely, all eigenvalues of the adjacency matrix $A(\Gamma)$ are rational integers. In section 3 we explain which subset $S$ of $G$ can be chosen to make $\Gamma$ being integral graph. We also compute $\alpha_{\chi}=\sum_{g\in S}\chi(g)$ for all character $\chi$ of $G$ which involved in our characterization on graph $\Gamma$ having PST in Theorem \ref{thm-main1} and \ref{thm-main}. Then we focus cubelike graphs in section 4. We find several new results including to answer problems raised in \cite{6, 15, 16}, we provide a lower bound of the minimum time $t$ for which a simple connected cubelike graph $\Gamma$ to have PST between two distinct vertices at the time $t$ (Lemma \ref{lem-4.2}), it is shown that the lower bound is tight when $n$ is odd by using bent functions. We also get several results for some new classes of abelian Cayley graphs. Section 5 is conclusion.
\section{Main results}

In this section, we present a characterization on connected simple abelian Cayley graphs having PST. From now on, let $G$ be a finite (additive) abelian group, $|G|=n\geq 2$, $S$ be a subset of $G$, $|S|=d\geq 1$. The Cayley graph $\Gamma={\rm Cay}(G,S)$ is defined by
\begin{eqnarray*}
  V &=& G, \mbox{ the set of vertices} \\
  E &=& \{(u,v): u, v\in G, u-v\in S\}, \mbox{ the set of edges}.
\end{eqnarray*}
We assume that $0\not\in S$ and $-S=S$ (which means that $\Gamma$ is a simple graph) and $G=\langle S\rangle$ ($G$ is generated by $S$ which means that $\Gamma$ is connected). The adjacency matrix of $\Gamma$ is defined by
$A=A(\Gamma)=(a_{g,h})_{g,h\in G}$ where
\begin{equation*}
  a_{g,h}=\left\{\begin{array}{cc}
                   1 & \mbox{ if $g-h \in S$} \\
                   0 & \mbox{ otherwise}.
                 \end{array}
  \right.
\end{equation*}

Let $\hat{G}$ be the character group of $G$. It is known that the eigenvalues of $A$ are the real numbers
\begin{equation*}
  \alpha_\chi=\sum_{g\in G}\chi(g), \ \ (\chi \in \hat{G}).
\end{equation*}
Each finite abelian group $G$ can be decomposed as a direct sum of cyclic groups
\begin{equation*}
  G=\mathbb{Z}_{n_1}\oplus \cdots\oplus \mathbb{Z}_{n_r}\ \ (n_s\geq 2)
\end{equation*}
where $\mathbb{Z}_m=(\mathbb{Z}/m\mathbb{Z},+)$ is a cyclic group of order $m$. For every $x=(x_1,\cdots,x_r)\in G$ $(x_s\in \mathbb{Z}_{n_s})$, the mapping
\begin{equation*}
  \chi_x:G\rightarrow \mathbb{C}, \chi_x(g)=\prod_{s=1}^r\omega_{n_s}^{x_sy_s} \ (\mbox{ for $g=(g_1,\cdots,g_r)\in G$})
\end{equation*}
is a character of $G$ where $\omega_{n_s}=\exp(2\pi i/n_s)$ is a primitive $n_s$-th root of unity in $\mathbb{C}$ and $e={\rm lcm}(n_1,\cdots, n_r)$ is the exponent of $G$ which is denoted by $\exp(G)$. Moreover, the mapping $G\rightarrow \hat{G}, x\mapsto \chi_x$ is an isomorphism of groups. Thus $\hat{G}=\{\chi_x| x \in G\}$ and $\chi_x(g)=\chi_g(x)$ for all $x,g\in G$.

In order to compute the transfer matrix $H(t)=\exp(itA)=(H_{u,v}(t))_{u,v\in G}$ we need to diagonalize the matrix $A=A(\Gamma)$.
Consider the $n\times n$ matrix
\begin{equation*}
  P=\frac{1}{\sqrt{n}}(c_{g,h})_{g,h\in G}, n=|G|, c_{g,h}=\chi_g(h).
\end{equation*}
By the orthogonal relation of characters we know that $P$ is a unitary matrix, $PP^*=I_n=P^*P$, where $P^*$ means the transpose of the conjugate of $P$.
Let $D$ be the following diagonal matrix
\begin{equation*}
  D={\rm diag}(\alpha_g: g\in G)=(d_{g,h}), d_{g,h}=\alpha_g\delta_{g,h}
\end{equation*}
where $\delta_{g,h}=1$ if $g=h$ and $0$ otherwise. Let $A=(a_{g,h})$ be the adjacency matrix of $\Gamma$ and let $AP=(\eta_{g,h}), PD=(\nu_{g,h})$. Then for every $g,h\in G$,
\begin{eqnarray*}
  \eta_{g,h} &=& \frac{1}{\sqrt{n}}\sum_{k\in G}a_{g,k}c_{k,h} =  \frac{1}{\sqrt{n}}\sum_{k\in G, g-k\in S}\chi_h(k)\\
   &=&\frac{1}{\sqrt{n}}\sum_{k'\in S}\chi_h(g-k')=\frac{\chi_h(g)}{\sqrt{n}}\sum_{k'\in G}\overline{\chi_h(k')}\\
   &=&\frac{\chi_h(g)}{\sqrt{n}}\overline{\alpha_h}=\frac{\chi_h(g)}{\sqrt{n}}{\alpha_h}.
\end{eqnarray*}
\begin{equation*}
  \nu_{g,h}=\frac{1}{\sqrt{n}}\sum_{k\in G}c_{g,k}d_{k,h}=\frac{1}{\sqrt{n}}\sum_{k\in G}\chi_{g}(k)\alpha_k\delta_{k,h}=\frac{\chi_h(g)}{\sqrt{n}}\alpha_h.
\end{equation*}
Thus $P^*AP=D$ and
\begin{equation*}
  H(t)=\exp(itA)=P\exp(itD)P^*=P\cdot {\rm diag}(\exp(it\lambda_g): g\in G)\cdot P^*=(H_{g,h}(t))
\end{equation*}
Therefore
\begin{eqnarray*}
  H_{g,h}(t)&=&\frac{1}{n}\sum_{x,y\in G}c_{g,x}\exp(it\alpha_x)\delta_{x,y}\overline{c_{h,y}}\\
  &=&\frac{1}{n}\sum_{x\in G}\exp(it\alpha_x)\chi_g(x)\overline{\chi_h(x)}\\
  &=&\frac{1}{n}\sum_{x\in G}\exp(it\alpha_x)\chi_a(x)
\end{eqnarray*}
where $a=g-h$. Therefore,
\begin{equation*}
  |H_{g,h}(t)|=1 \mbox{ if and only if } |\sum_{x\in G}\exp(it\alpha_x)\chi_a(x)|=n.
\end{equation*}
Since
$\alpha_x$ is a real number and $|\exp(it\alpha_x)\chi_a(x)|=1$ for all $x\in G$. The summation has $n=|G|$ terms, we know that $|H_{g,h}(t)|=1$ if and only if all $\exp(it\alpha_x)\chi_a(x)$ ($x\in G$) are the same number. For $x=0$, we have $\alpha_0=\sum_{g\in S}\chi_0(g)=|S|=d$, $\chi_a(0)=1$ and $\exp(it\alpha_0)\chi_a(0)=\exp(itd)$. Thus we obtain the following preliminary result:

\begin{lem}\label{lem-2} Let $\Gamma={\rm Cay}(G,S)$ be an abelian Cayley simple graph, $d=|S|$. For $g,h\in G$, let $a=g-h$. Then the following statements are equivalent:

(1) $\Gamma$ has a PST between vertices $g$ and $h$ at the time $t>0$;

(2) For any $x(\neq 0)\in G$, $\chi_a(x)=\exp(it(d-\alpha_x))$.\end{lem}

As direct consequences of Lemma \ref{lem-2}, we get the following two remarkable necessary conditions on an abelian Cayley graph $\Gamma={\rm Cay}(G,S)$ having PST. Firstly, the eigenvalues $\alpha_x(x\in G)$ of $A=A(\Gamma)$ should be rational integers. We call such graph as integral graph. This result has been given in \cite[Theorem 3.5]{23}. In fact, it can be derived from more general fact that any vertex-transitive graph having PST is integral, but our proof is more simple and straightforward. Secondly, if $\Gamma$ has PST between two distinct vertices $g$ and $h$, then the order of $a=g-h$ should be two. This result has been given for circulant and cubelike graphs in \cite{5} and \cite{6} respectively.


\begin{lem}\label{lem-3} Let $\Gamma={\rm Cay}(G,S)$ be an abelian simple Cayley graph, $n=|G|\geq 3$. Assume that $\Gamma$ has a PST between pair $(g,h)$ of vertices. Then

(A) $\Gamma$ is an integral graph. Namely, $\alpha_x\in \mathbb{Z}$ for all $x\in G$,

(B) If $a=g-h\neq 0$, then the order of $a$ is two. Consequently, $|G|=n$ is even.

\end{lem}

\begin{proof} Suppose that $\Gamma$ has a PST between $g$ and $h\in G$. By Lemma \ref{lem-2}, $$\chi_a(x)=\exp(it(d-\alpha_x))$$
for all $0\neq x\in G$, where $d=|S|$ and $a=g-h$. Let $m$ be the order of $a$. Then the order of $\chi_a$ is also $m$. Thus we can write
\begin{equation*}
  \chi_a(x)=\omega_m^{i_a(x)}, \mbox{ where $\omega_m=\exp(2\pi i/m), i_a(x)\in \mathbb{Z}_m$}.
\end{equation*}
Then the condition (2) of Lemma \ref{lem-2} becomes that
\begin{equation}\label{f-4}
  \exp(it(d-\alpha_x))=\exp(\frac{2\pi i}{m}i_a(x)).
\end{equation}
 (A) Let $t=2\pi T$. We get
\begin{equation}\label{f-5}
  M_x=:T(d-\alpha_x)-\frac{i_a(x)}{m}\in \mathbb{Z} \mbox{ for any $0\neq x\in G$},
\end{equation}
Thus $M=:\sum_{0\neq x\in G}M_x\in \mathbb{Z}$. On the other hand,
\begin{equation*}
  M=\sum_{0\neq x\in G}(T(d-\alpha_x)-\frac{i_a(x)}{m})=(n-1)Td-T\sum_{0\neq x\in G}\alpha_x-\frac{1}{m}\sum_{0\neq x\in G}i_a(x),
\end{equation*}
and
\begin{eqnarray*}
 \sum_{0\neq x\in G}\alpha_x&=&\sum_{0\neq x\in G}\sum_{g\in S}\chi_x(g)=\sum_{g\in S}\sum_{0\neq x\in G}\chi_g(x)=-d.
\end{eqnarray*}
Thus
$M=(n-2)dT-\frac{1}{m}\sum_{0\neq x\in G}i_a(x)\in \mathbb{Z}.$ By assumption $n=|G|\geq 3$ and $d, i_a(x)\in \mathbb{Z}$, we get that $T\in \mathbb{Q}$ (the field of rational numbers.) Then by $T>0$ and formula (\ref{f-5}), we know that $\alpha_x\in \mathbb{Q}$. Since $\alpha_x$ is an algebraic integer, we get $\alpha_x\in \mathbb{Z}$ for all $0\neq x\in G$ and $\alpha_0=d$. Therefore $\Gamma$ is an integral graph.


(B) Suppose that $a=g-h(\neq 0)$ and the order of $a$ is $m\geq 2$. Then $\chi_a(x)=\chi_x(a)$ and the order of $\chi_a$ is $m$ so that here exists an element $x\in G$ such that $\chi_a(x)=\omega_m^{i_a(x)}$ and $\gcd(i_a(x),m)=1$. Obviously, $x$ should be non-zero. By (\ref{f-5}), we have
\begin{equation}\label{f-6}
 T(d-\alpha_x)-\frac{i_a(x)}{m}\in \mathbb{Z}.
\end{equation}
Since $\alpha_x$ is real, we have
\begin{equation*}
  \alpha_{-x}=\sum_{g\in S}\chi_{-x}(g)=\overline{\alpha_x}=\alpha_x,
\end{equation*}
and
\begin{equation*}
  \omega_m^{i_a(-x)}=\chi_a(-x)=\overline{\chi_a(x)}=\omega_m^{-i_a(x)}, i_a(-x)=-i_a(x)
\end{equation*}
Thus by (\ref{f-6}), we get
\begin{equation}\label{f-7}
   T(d-\alpha_x)+\frac{i_a(x)}{m}\in \mathbb{Z}.
\end{equation}
Combining (\ref{f-6}) and (\ref{f-7}) together, we have $2i_a(x)/m\in \mathbb{Z}$, since $\gcd(i_a(x),m)=1$, we get that $m=2$.
This completes the proof of Lemma \ref{lem-3}.\end{proof}

Now we can present a characterization on periodicity of integral connected abelian Cayley graphs. Let $\Gamma$ be such a graph. Since $\Gamma$ is integral, $\alpha_x\in \mathbb{Z}$ for all $x\in G$. We define
\begin{equation}\label{f-2.5}
  M=\gcd(d-\alpha_x: x\in G) \ \ (d=|S|)
\end{equation}
where $d-\alpha_x=d-\sum_{g\in S}\chi_x(g)\geq 0$ and $d-\alpha_0=0$. For any $x\in G$, if $d-\alpha_x=0$ then $\chi_x(g)=1$ for all $g\in S$. Since $\Gamma$ is connected, $G=\langle S\rangle$, we know that $\chi_x(g)=1$ for all $g\in G$ which implies that $x=0$. In other words, $d-\alpha_x$ is a positive integer for any $0\neq x\in G$. Therefore if $n=|G|\geq 2$, then $G$ has non-zero element $x$ and $d-\alpha_x>0$. Therefore $M$ is a positive integer.

\begin{thm}\label{thm-main1} Let $\Gamma={\rm Cay}(G, S)$ be an integral connected simple abelian Cayley graph, $n=|G|\geq 2$  and $d=|S|$. Then for every vertex $g\in G$, $\Gamma$ is periodic at vertex $g$ and the set
\begin{equation*}
  T(g,g)=\{t>0| g \mbox{ is periodic with period $t$}\}
\end{equation*}
is $\{\frac{2\pi l}{M}|l=1,2,\cdots\}$, where $M$ is defined by (\ref{f-2.5}).\end{thm}

\begin{proof}Since $a=g-g=0$ and $i_x(a)=0$ for each $x\in G$, from formula (\ref{f-6}) we know that $g$ is periodic with period $t=2\pi T>0$ if and only if
\begin{equation}\label{f-new1}
  T(d-\alpha_x)\in \mathbb{Z}, \mbox{ for all $x\in G$}.
\end{equation}
It is easy to see that (\ref{f-new1}) is equivalent to $TM\in \mathbb{Z}$. Therefore
\begin{equation*}
  T(g,g)=\left(\frac{2\pi}{M}\mathbb{Z}\right)\cap \mathbb{R}_{>0}=\frac{2\pi}{M}\mathbb{Z}_{\geq 1}.
\end{equation*}
\end{proof}

Now we consider the case $g,h\in G$ and $a=g-h\neq 0$. By Lemma \ref{lem-3}, if $\Gamma$ has PST between $g$ and $h$, then $\alpha_x=\sum_{g\in S}\chi_x(g)\in \mathbb{Z}$ for all $x\in G$ and the order of $a$ is two. Therefore $n=|G|$ is even and for any $x\in G$, we have $\chi_a(x)=\pm 1$. Let $\chi_a(x)=(-1)^{m_a(x)}$, $m_a(x)\in \{0,1\}$ and
\begin{equation}
\begin{array}{ll}
  &G_\varepsilon = G_{\varepsilon,a}=\{x\in G: m_a(x)=\varepsilon\}=\{x\in G: \chi_a(x)=(-1)^\varepsilon\}, (\varepsilon \in \{0,1\}) \\
&M_\varepsilon= M_{\varepsilon,a}=\gcd(d-\alpha_x: x\in G_\varepsilon)
\end{array}\label{f-new2}
\end{equation}
  Then $G_0=\{x\in G: \chi_a(x)=1\}$ is a subgroup of $G$, $G=G_0\cup G_1$ and $|G_0|=\frac{n}{2}$. If $n=|G|\geq 4$, both of $G_0$ and $G_1$ have non-zero element $x$, $d-\alpha_x\neq 0$. Thus both of $M_0$ and $M_1$ are positive integers.

  We also need notation of the $2$-adic exponential valuation of rational numbers which is a mapping defined by
\begin{equation*}
  v_2:\mathbb{Q}\rightarrow \mathbb{Z}\cup \{\infty\}, v_2(0)=\infty, v_2(2^\ell\frac{a}{b})=\ell, \mbox{ where $a,b,\ell\in \mathbb{Z}$ and $2\not |ab$}.
\end{equation*}
We assume that $\infty+\infty=\infty+\ell=\infty$ and $\infty>\ell$ for any $\ell \in \mathbb{Z}$. Then $v_2$ has the following properties. For $\beta,\beta'\in \mathbb{Q}$,

(P1) $v_2(\beta\beta')=v_2(\beta)+v_2(\beta')$;

(P2) $v_2(\beta+\beta')\geq \min(v_2(\beta),v_2(\beta'))$ and the equality holds if $v_2(\beta)\neq v_2(\beta')$.

We present our main result of this paper.

\begin{thm}\label{thm-main}Let $\Gamma={\rm Cay}(G,S)$ be a connected simple abelian Cayley graph, $n=|G|\geq 3$, $d=|S|$. Then for $g,h\in G, a=g-h\neq 0$, $\Gamma$ has PST between $g$ and $h$ if and only if the following three conditions hold:

(I) $\Gamma$ is a integral graph. Namely., the eigenvalues $\alpha_x=\sum_{g\in G}\chi_x(g)\in \mathbb{Z}$ for all $x\in G$;

(II) the order of $a$ is two;

(III) $v_2(d-\alpha_x)$ for all $x\in G_1$ are the same number, say, $\rho$, and $v_2(M_0)\geq \rho+1$, where $M_0$ is defined by (\ref{f-new2}).

Moreover, if the conditions (I)-(III) are satisfied, then the set
\begin{equation*}
  T(g,h)=\{t>0|\mbox{ $\Gamma$ has a PST between $g$ and $h$ at time $t$}\}
\end{equation*}
is $\{\frac{\pi}{M}+\frac{2\pi}{M}\ell:\ell=0,1,2,\cdots\}$, where $M=\gcd(d-\alpha_x:0\neq x\in G)$.
\end{thm}

\begin{proof}The conditions (I), (II) are given by Lemma \ref{lem-3}. Then by formula (\ref{f-5}), $\Gamma$ has PST between $g$ and $h$ at the time $t=2\pi T>0$ if and only if the following two conditions hold:

(i) $T(d-\alpha_x)\in \mathbb{Z}$ for all $x\in G_0$;

(ii) $T(d-\alpha_x)-\frac{1}{2}\in \mathbb{Z}$ for all $x\in G_1$;

Condition (i) means that $T\in \frac{1}{M_0}\mathbb{Z}=\{\frac{\ell}{M_0}|\ell\in \mathbb{Z}\}$.
 Now, we consider condition (ii). Suppose that $T(d-\alpha_x), T(d-\alpha_x')\in \frac{1}{2}+\mathbb{Z}$, then $T\in \mathbb{Q}\setminus \{0\}$ and $v_2(T(d-\alpha_x))=v_2(T(d-\alpha_x'))=-1$. Therefore
$v_2(d-\alpha_x)=v_2(d-\alpha_x')=-1-v_2(T)$. $v_2(d-\alpha_x)$ for all $x\in G_1$ are the same number $\rho=-1-v_2(T)\geq 0$. Furthermore, if $v_2(d-\alpha_x)=\rho$ for all $x\in G_1$, then $v_2(M_1)=\rho$ and $v_2(T)=v_2(T(d-\alpha_x))-v_2(d-\alpha_x)=-(\rho+1)$. Thus condition (ii) means that
$T\in \frac{1}{M_1}(\frac{1}{2}+\mathbb{Z})=\{\frac{1}{M_1}(\frac{1}{2}+\ell):\ell \in \mathbb{Z}\}$. Putting conditions (i) and (ii) together, we get
\begin{eqnarray}
\nonumber  T(g,h) &=& \left(\frac{2\pi }{M_0}\mathbb{Z}\right)\cap \left(\frac{2\pi }{M_1}(\frac{1}{2}+\mathbb{Z})\right) \cap \mathbb{R}_{>0}\\
\label{f-9}   &=& \frac{\pi}{M_0M_1}(2M_1\mathbb{Z}\cap M_0(1+2\mathbb{Z}))\cap \mathbb{R}_{>0}.
\end{eqnarray}
For every $z\in \mathbb{Z}$, it is easy to check that
\begin{eqnarray*}
& &z\in 2M_1\mathbb{Z}\cap M_0(1+2\mathbb{Z}) \\
   &\Leftrightarrow & z=2M_1x=M_0(1+2y) \mbox{ for some $x,y\in \mathbb{Z}$}\\
   &\Leftrightarrow&2M_1x-2M_0y=M_0\mbox{ has solution $x,y\in \mathbb{Z}$}\\
   &\Leftrightarrow& \gcd(2M_0,2M_1)|M_0\\
   &\Leftrightarrow&v_2(M_0)\geq \rho+1 \mbox{ (since $v_2(M_1)=\rho$)}.
\end{eqnarray*}
Therefore, if $\Gamma$ has a PST between $g$ and $h$, then condition (I)-(III) hold. On the other hand, suppose that conditions (I)-(III) hold, let $M=\gcd(M_0,M_1)$, where $M_0,M_1$ is defined by (\ref{f-new2}). Write $M_0=Mm_0, M_1=Mm_1$. Then $\gcd(m_0,m_1)=1$. From $v_2(M_1)=\rho$, $v_2(M_0)\geq \rho+1$, we get that $v_2(M)=\rho$ and $m_0$ is even, $m_1$ is odd. Then
$2M_1x-2M_0y=M_0\Leftrightarrow m_1x-m_0y=\frac{m_0}{2}\in \mathbb{Z}$. Since $\gcd(m_0,m_1)=1$, we know that the solutions of the Diophantine equation $m_1x-m_0y=\frac{m_0}{2}$ are
\begin{equation*}
  \left\{\begin{array}{cc}
           x= & \frac{m_0}{2}+m_0\ell \\
          y= & \frac{m_1-1 }{2}+m_1\ell.
         \end{array}
  \right. (\ell \in \mathbb{Z})
\end{equation*}
Thus
\begin{equation*}
  z=2M_1x=2M_1\frac{m_0}{2}(1+2\ell)=\frac{M_0M_1}{M}(1+2\ell),
\end{equation*}
and
\begin{equation*}
  2M_1\mathbb{Z}\cap M_0(1+2\mathbb{Z})=\frac{M_0M_1}{M}(1+2\mathbb{Z}).
\end{equation*}
By (\ref{f-9}), we get
\begin{equation*}
  T(g,h)=(\frac{\pi}{M}+\frac{2\pi}{M}\mathbb{Z})\cap \mathbb{R}_{>0}=\frac{\pi}{M}+\frac{2\pi}{M}\mathbb{Z}_{\geq 0}.
\end{equation*}
This completes the proof.
\end{proof}

For convenience of application later, we give several version of the condition (III) of theorem 2.4.

\begin{cor}\label{cor-2.5}Let $\Gamma={\rm Cay}(G,S)$ be a connected integral simple abelian Cayley graph, $n=|G|\geq 4$, $d=|S|$, $g,h\in G$ and $a=g-h$ is an element of order two. Let $G_\varepsilon (\varepsilon=0,1)$ be the partition of $G$ defined by (\ref{f-new2}). Let $K=K(G_0,G_1)$ be the complete bipartite graph on $\{G_0, G_1\}$ ($|G_0|=|G_1|=\frac{n}{2}$). Then the following conditions are equivalent to each others.

(1) $\Gamma$ has PST between vertices $g$ and $h$;

(2) $v_2(d-\alpha_y)=\rho$, $v_2(d-\alpha_x)\geq \rho+1$ for all $y\in G_1$ $x\in G_0$ for some integer $\rho\geq 0$;

(3) $v_2(\alpha_x-\alpha_y)=\rho$ for all $x\in G_0$, $y\in G_1$ for some integer $\rho\geq 0$;

(4) There exists a connected spanning subgraph $K'$ of $K$ such that for each edge $\{z,z'\}$ in $K'$, $v_2(\alpha_z-\alpha_{z'})=\rho$ for some integer $\rho\geq 0$.\end{cor}

\begin{proof}$(1)\Leftrightarrow (2)$: Since $(2)$ is the condition (III) of Theorem \ref{thm-main}.

$(2)\Rightarrow (3)$: By the property (P2) of the exponential valuation $v_2$ we know that if $v_2(d-\alpha_y)=\rho$ and $v_2(d-\alpha_x)=\rho+1$, then $v_2(\alpha_x-\alpha_y)=\min\{v_2(d-\alpha_y),v_2(d-\alpha_x)\}=\rho$.

$(3)\Rightarrow (4)$: This is obvious by choosing $K'=K$.

$(4)\Rightarrow (2)$: Suppose that the condition (4) holds. Then for each $y\in G_1$, there exists a path $\{x_0=0, y_0,x_1,y_1,\cdots,x_\ell, y_\ell=y\}$. By assumption of (4), we have $v_2(\alpha_{x_0}-\alpha_{y_0})=v_2(\alpha_{y_0}-\alpha_{x_1})=\rho$ which means that $\alpha_{x_0}-\alpha_{y_0}\equiv \alpha_{y_0}-\alpha_{x_1}\equiv 2^\rho ({\rm mod} \ 2^{\rho+1})$. Therefore
\begin{equation*}
  \alpha_{x_0}-\alpha_{x_1}=(\alpha_{x_0}-\alpha_{y_0})+(\alpha_{y_0}-\alpha_{x_1})\equiv 0 ({\rm mod} \ 2^{\rho+1}).
\end{equation*}
Namely, $v_2(\alpha_{x_0}-\alpha_{x_1})\geq \rho+1$. Then by assumption $v_2(\alpha_{x_1}-\alpha_{y_1})=\rho$ we get $v_2(\alpha_{x_0}-\alpha_{y_1})=\rho$. Continuing this process we get $v_2(d-\alpha_y)=v_2(\alpha_0-\alpha_y)=v_2(\alpha_{x_0}-\alpha_{y_{\ell}})=\rho$. In the same way we get $v_2(d-\alpha_x)\geq \rho+1$ for each $x\in G_0$. This completes the proof of Corollary \ref{cor-2.5}.
\end{proof}

In circulant case, $G=\mathbb{Z}_n$, $n=2m, m\geq 1$. There is a unique element $a$ with order two in $G$. The character group of $G$ is $\widehat{G}=\{\chi_x:x\in G\}$ where $g\in G$, $\chi_x(g)=\omega_n^{xg}$. Moreover,
\begin{eqnarray*}
  G_0 &=& \{x\in G: \chi_a(x)=1\}=\{x\in G: 2|x\} \\
  G_1 &=& \{x\in G: \chi_a(x)=-1\}=\{x\in G: 2\nmid x\}.
\end{eqnarray*}
The path $``0-1-2-\cdots-(n-1)"$ is a connected spanning subgraph of the complete bipartite graph $K(G_0, G_1)$. By Corollary \ref{cor-2.5} (4), we get the following characterization of a circulant graph in \cite{6}.

\begin{cor}\label{cor-2.7}Let $G=\mathbb{Z}_n$, $n=2m$, $S\subseteq G$ and $\Gamma={\rm Cay}(G,S)$ be an integral connected (circulant) Cayley graph. Then $\Gamma$ has PST between $g$ and $g+m$ (for each $g\in G$) if and only if there exists $\rho \in \mathbb{Z}$ such that
\begin{equation*}
  v_2(\alpha_j-\alpha_{j+1})=\rho \ \ (0\leq  j\leq n-2)
\end{equation*}
where $\alpha_j=\chi_j(S)=\sum_{g\in S}\omega_n^{jg}$.
\end{cor}

\begin{remark}Ba${\rm \check{s}}$i${\rm \acute{c}}$ \cite{4} has presented a necessary and sufficient condition on $S$ such that the integral connected circulant Cayley $\Gamma={\rm Cay}(\mathbb{Z}_{2m},S)$ has PST between $g$ and $g+m$.\end{remark}

\section{A characterization of abelian Cayley graph being integral}

In this section we introduce a characterization on subset $S$ of a finite abelian group $G$ such that the graph $\Gamma={\rm Cay}(G,S)$ is integral. This topic has been discussed in \cite{20}. But, as indicated by C. Godisl \cite{16}, the characterization has been given in (\cite{9}, 1982) involved in more general problems on abelian groups.

From now on we denote $\chi(S)=\sum_{g\in S}\chi(g)$ for each $\chi\in \widehat{G}$ and subset $S$ of $G$. Let $e=\exp(G)$ be the exponent of $G$, then for each $z\in G$, $\alpha_z=\chi_z(S)$ belongs to the ring $\mathbb{Z}[\omega_e]$ of integers of the cyclotomic field $K=\mathbb{Q}(\omega_e)$. The Galois group of the extension $K/\mathbb{Q}$ is
\begin{equation*}
  {\rm Gal}(K/\mathbb{Q})=\{\sigma_\ell: \ell \in \mathbb{Z}^*\}\cong \mathbb{Z}_e^*
\end{equation*}
where $\mathbb{Z}_e^*=\{\ell \in \mathbb{Z}_e: \gcd(\ell,e)=1\}$ is the unit group of the ring $\mathbb{Z}_e=\mathbb{Z}/e\mathbb{Z}$ and $\sigma_{\ell}$ is defined by $\sigma_{\ell}(\omega_e)=\omega_e^\ell$. Therefore the Cayley graph $\Gamma={\rm Cay}(G,S)$ is integral if and only if $\alpha_z=\chi_z(S)\in \mathbb{Z}$ for each $z\in G$, if and only if $\sigma_\ell(\chi(S))=\chi(S)$ for each $\chi\in \hat{G}$ and $\ell \in \mathbb{Z}_e^*$. But $\sigma_\ell(\chi(S))=\sigma_\ell(\sum_{g\in S}\chi(g))=\sum_{g\in S}\sigma_\ell(\chi(g))=\sum_{g\in S}\chi(g)^\ell=\sum_{g\in S}\chi(\ell g)$. Therefore, $\Gamma$ is integral if and only if $\chi(S)=\chi(\ell S)$ for all $\chi \in \hat{G}$ and $\ell \in \mathbb{Z}_e^*$, $\ell S=:\{\ell g: g\in S\}$. Then by the orthogonal relation of characters we get that

\begin{lem}\label{lem-3.1}The graph $\Gamma={\rm Cay}(G,S)$ is integral if and only if $\ell S=S$ for all $\ell \in \mathbb{Z}_e^*$.\end{lem}

From Lemma \ref{lem-3.1}, it is natural to make the following definition.

\begin{defn}\label{def-3.2} Let $(G,+)$ be a finite abelian group, $e=\exp(G)$. Two elements $g,h\in G$ are called $\mathbb{Z}_e^*$-equivalent if there exists $\ell \in \mathbb{Z}_e^*$ such that $g=\ell h$. Namely, $g$ and $h$ generates the same (cyclic) subgroup $\langle g \rangle=\langle h\rangle$ of $G$.\end{defn}

This is an equivalent relation on $G$ and $G$ is partitioned into $\mathbb{Z}_e^*$-classes. For $g\in G$, we denote the class of $g$ by $[g]$. If the order of $g$ is $\lambda$ ($\lambda|e$), then the $\mathbb{Z}_e^*$-class $[g]$ is the $\mathbb{Z}_\lambda^*$-class, and the size of $[g]$ is $\varphi(\lambda)$ where $\varphi$ is the Euler function which is the number of $t\in \mathbb{Z}$ such that $1\leq t\leq \lambda$ and $\gcd(t, \lambda)=1$. As a direct consequence of Lemma \ref{lem-3.1} we get

\begin{thm}(\cite{9,20})\label{thm-3.3} Let $G$ be a finite abelian group, $S\subseteq G$, $e=\exp(G)$. The Cayley graph $\Gamma={\rm cay}(G,S)$ is integral if and only if $S$ is an (disjoint) union of several $\mathbb{Z}_e^*$-classes of $G$.\end{thm}

From now on we call any union of $\mathbb{Z}_e^*$-classes of $G$ as $\mathbb{Q}$-set.

Any finite abelian group $G$ can be expressed as a direct sum of cyclic subgroups.
\begin{equation*}
  G=\mathbb{Z}_{n_1}\oplus\mathbb{Z}_{n_2}\oplus\cdots\oplus\mathbb{Z}_{n_r}, n=|G|=n_1n_2\cdots n_r.
\end{equation*}
We denote $\bar{n}=(n_1,n_2,\cdots,n_r)$. For any $\bar{d}=(d_1,d_2,\cdots,d_r)$, $d_i\geq 1$. If $d_j|n_j(1\leq j\leq r)$, we denote $\bar{d}|\bar{n}$. Let
\begin{equation*}
  D(\bar{n})=\{\bar{d}:\bar{d}|\bar{n}, d_1d_2\cdots d_r<n\}.
\end{equation*}
For any $\bar{d}\in D(\bar{n})$, let
\begin{equation*}
  \Sigma(\bar{d})=\{\bar{k}=(k_1,k_2,\cdots, k_r)\in G, k_j\in \mathbb{Z}_{n_j}, \gcd(k_j,n_j)=d_j \ \ (1\leq j\leq r)\}.
\end{equation*}
It is easy to see that $\sum(\bar{d})$ is a $\mathbb{Q}$-set of $G$. More general, for any subset $D$ of $D(\bar{n})$, the subset $\sum(D)=\cup_{\bar{d}\in D}\sum(\bar{d})$ of $G$ is a $\mathbb{Q}$-set, called gcd-set. The integral Cayley graph  $\Gamma={\rm Cay}(G,\sum(D))$ is called gcd-graph. The periodicity and PST on gcd-graphs have been researched in \cite{4, 5, 6,24,25}. For any circulant Cayley graph, $G$ is cyclic $(r=1)$ and any $\mathbb{Q}$-set of $G$ is a gcd-set. But if $G$ is not cyclic, there exists $\mathbb{Q}$-set of $G$ which may not be gcd-set as shown in the following example.

\begin{example}Let $G=\mathbb{Z}_4\oplus \mathbb{Z}_4$ (direct sum). Then $G$ has the following $\mathbb{Z}_4^*$-classes:
  $$\begin{array}{l}
    S_0=\{(0,0)\}, S_1=\{(0,2)\}, S_2=\{(2,0)\}, S_3=\{(2,2)\} \\
    S_4=\{(1,0),(3,0)\}, S_5=\{(1,2),(3,2)\}, S_6=\{(1,1), (3,3)\} \\
    S_7=\{(1,3),(3,1)\}, S_8=\{(0,1),(0,3)\}, S_9=\{(2,1),(2,3)\} \\
  \end{array}$$
where both of $S_6$ and $S_7$ are not gcd-set, but $S_6\cup S_7$ is the gcd-set $\sum(\bar{d})$, $\bar{d}=(1,1)$.

By Theorem \ref{thm-3.3}, any $\mathbb{Q}$-set $S$ of $G$ ($0\not\in S$) is an union of several classes $S_j$ ($1\leq j\leq 9$). Let $I$ be a subset of $\{1,2,\cdots,9\}$, $S=\cup_{j\in I}S_j$. It is easy to see that $-S_j=S_j$ so that $S=-S$ and $\Gamma={\rm Cay}(G,S)$ is an integral simple graph. The eigenvalues are
\begin{equation*}
  \alpha_z=\chi_z(S)=\sum_{j\in I}\chi_z(S_I),
\end{equation*}
where for $z=(z_1,z_2)\in G$ and $s=(s_1,s_2)\in G$, $\chi_z(s)=\omega_4^{z_1s_1+z_2s_2}$. Let $d_j=|S_j|$ ($1\leq j\leq |I|$), then $d=|S|=\sum_{j\in I}d_j$.

Now we take $a=(0,2)$, an element of $G$ with order two. Then
\begin{eqnarray}
  G_0 &=& \{x\in G: \chi_a(x)=1\}=\{x=(x_1,x_2): x_1,x_2\in \mathbb{Z}_4, 2|x_2\} \\
  G_1&=& \{y\in G: \chi_a(y)=-1\}=\{y=(y_1,y_2): y_1,y_2\in \mathbb{Z}_4, 2\nmid y_2\}.
\end{eqnarray}
By Corollary \ref{cor-2.5}, $\Gamma$ has PST between $g$ and $g+a$ $(g\in G)$ if and only if there exists $\rho\geq 0$ such that
\begin{equation}\label{f-e1}
   v_2(d-\alpha_y)=\rho, \mbox{ and } v_2(d-\alpha_x)\geq \rho+1 \mbox{ for all $x\in G_0, y\in G_1$}.
\end{equation}
In order to compute $\alpha_z=\sum_{j\in I}\chi_z(S_j)$, we need to determine all $\chi_z(S_j)=\sum_{g\in S_j}\chi_z(g)\in \mathbb{Z}$ for all $z\in G$ and $j\in I$. Moreover, if $z$ and $z'$ in the same class, $z=\ell z'$, $\ell \in \mathbb{Z}_4^*=\{1,3\}$, then $\chi_z(S_j)=\chi_{\ell z'}(S_j)=\chi_{z'}(\ell S_j)=\chi_{z'}(S_j)$ (since $\ell S_j=S_j$). Therefore, if we choose $z_j\in S_j$ ($1\leq j\leq 9$), it is enough to compute $\chi_{z_s}(S_j)$ ($1\leq j,s\leq 9$). We show the result of computation in the following table
\begin{center}
\begin{tabular}{|c|c|c|c|c|c|c|c|c|c|}
  \hline
  $z$ & $\chi_z(S_1)$ & $\chi_z(S_2)$&$\chi_z(S_3)$ &$\chi_z(S_4)$ & $\chi_z(S_5)$ & $\chi_z(S_6)$& $\chi_z(S_7)$ & $\chi_z(S_8)$ & $\chi_z(S_9)$ \\
  \hline
  $z_0=(0,0)$ & 1 & 1 & 1 & 2 & 2 & 2 & 2 & 2& 2 \\
  $z_1=(0,2)$ & 1 & 1 & 1 & 2 & 2 & -2 & -2 & -2 & -2 \\
  $z_2=(2,0)$ & 1& 1 & 1 & -2 & -2 & -2 & -2 & 2 & 2 \\
  $z_3=(2,2)$& 1 & 1 & 1 & -2 & -2 & 2 & 2 &-2 & -2\\
  $z_4=(1,0)$& 1 & -1 & -1 & 0 & 0 & 0 & 0& 2& -2 \\
  $z_5=(1,2)$ & 1 & -1 & -1 & 0 & 0& 0 & 0 & -2 & 2 \\
  $z_6=(1,1)$ & -1 & -1 &1 & 0 & 0 & -2 & 2 & 0& 0 \\
  $z_7=(1,3)$& -1 & -1 & 1 & 0 & 0 & 2 & -2& 0 & 0 \\
  $z_8=(0,1)$ & -1 & 1 & -1 & 2 & -2 & 0 & 0 & 0 & 0 \\
 $z_9=(2,1)$ & -1 & 1 & -1 & -2 & 2 & 0 & 0 & 0 & 0 \\
  \hline
\end{tabular}\end{center}
It is easily seen that $G_0=\cup_{j=0}^5S_j$ and $G_1=\cup_{j=6}^9S_j$. There are many $\mathbb{Q}$-set $S$ such that $\Gamma$ has PST between $g$ and $g+a$. One of them is $S=S_4\cup S_6$ (which is not a gcd-set). It is easy to check that $\langle S\rangle=G$ so that $\Gamma$ is a connected integral graph. From the table we get
\begin{equation*}
  (\chi_0(S),\cdots,\chi_5(S))=(\chi_0(S_4)+\chi_0(S_6),\cdots,\chi_5(S_4)+\chi_5(S_6))=(4,0,-4,0,0,0),
\end{equation*}
$d=|S_4|+|S_6|=4$ and $(d-\chi_0(S),\cdots,d-\chi_5(S))=(0,4,8,4,4,4)$. Thus $\min\{v_2(d-\chi_x(S)):x\in G_0\}\geq 2$. On the other hand, $(\chi_6(S),\cdots,\chi_9(S))=(-2,2,2,-2)$ and $(d-\chi_6(S),\cdots,d-\chi_9(S))=(6,2,2,6)$. Thus for each $y\in G_1$, $v_2(d-\alpha_y)=v_2(d-\chi_j(S))=1$, $j\in \{6,7,8,9\}$. By (\ref{f-e1}) ($\rho=1$) we know that $\Gamma$ has PST between $g$ and $g+a$ ($a=(0,2)$) for any $g\in G$.
\end{example}

At the end of this section we present a non-existence result on PST in abelian Cayley graphs $\Gamma={\rm Cay}(G,S)$, we have shown that if $\Gamma$ has PST between two distinct vertices $g$ and $h$ in $G$, then the order of $a=g-h$ should be even so that $n=|G|$ is even. For circulant case, the stronger necessary condition $4|n$ has been proved in \cite{25} by heavy computations. Now we show that this is true for any abelian Cayley graph.

\begin{thm}\label{thm-3.5}Let $\Gamma={\rm Cay}(G,S)$ be an integral abelian Cayley graph, $n=|G|\equiv 2 ({\rm mod } \ 4)$ and $n\geq 6$. Then $\Gamma$ has no PST between any distinct vertices.\end{thm}

\begin{proof}By the assumption and $n=2N, 2\nmid N$, $G$ is a direct sum of two subgroups $G=\mathbb{Z}_2\oplus H$, where $|H|=N$, and $a=(1,0)$ is the unique element of order two in $G$. The character group of $G$ is
\begin{equation*}
  \hat{G}=\hat{\mathbb{Z}_2}\otimes \hat{H} (\mbox{ direct product })=\{\chi,\varphi \chi|\chi\in \hat{H}\}
\end{equation*}
where for $g=(b,h)\in \mathbb{Z}_2\oplus H$, $\chi(g)=\chi(h), \varphi(g)=(-1)^b$, and there exists the canonical isomorphism of groups
\begin{equation*}
  H\rightarrow \hat{H}, h\mapsto \chi_h.
\end{equation*}
Then by $a=(1,0)$,
\begin{eqnarray*}
       G_0 &=& \{g=(b,h)\in G: (-1)^b\chi_h(0)=1\}=(0,H) \\
       G_1 &=& \{g=(b,h)\in G: (-1)^b\chi_h(0)=-1\}=(1,H).
     \end{eqnarray*}
     Suppose that $\Gamma$ has PST between $g$ and $g+a$ in $G$. Let $S_\varepsilon=S\cap G_\varepsilon$ $(\varepsilon=0,1)$, $|S_\varepsilon|=d_\varepsilon$. Then
     \begin{equation*}
       S=S_0\cup S_1 \mbox{ and } d=|S|=d_0+d_1.
     \end{equation*}
By Corollary \ref{cor-2.5} (3), there exists $\rho\geq 0$ such that
\begin{equation}\label{f-3.4}
 v_2(\alpha_x-\alpha_y)=\rho \mbox{ for all $x\in G_0, y\in G_1$}.
\end{equation}
For any $h\in H$, $x=(0,h)\in G_0$, $y=(1,h)\in G_1$, we have
\begin{equation*}
  \alpha_x=\chi_h(S)=\chi_h(S_0)+\chi_h(S_1), \alpha_y=\varphi\chi_h(S)=\chi_h(S_0)-\chi_h(S_1).
\end{equation*}
It is easy to see that $S_0$ and $S_1$ are $\mathbb{Q}$-sets, thus $\chi_h(S_0), \chi_h(S_1)\in \mathbb{Z}$. By (\ref{f-3.4}), we get $v_2(2\chi_h(S_1))=\rho$. Namely
\begin{equation}\label{f-3.5}
  v_2(\chi_h(S_1))=\rho-1 \mbox{ for any $h\in H$}.
\end{equation}
Particularly, taking $h=0\in H$ we get $v_2(d_1)=\rho-1$. Moreover, $[(0,0)]$ and $[(1,0)]$ are classes with one element. The other element $z$ in $G$ has order $e\geq 3$, the size of $[z]$ is an even number $\varphi(e)$. Therefore $2|d_0$, and $2|d_1$ if and only if $a\not\in S$ (note that $S$ is a union of some classes and $(0,0)\not\in S$).

If $a\not\in S$, then $a\not\in S_1$ and $\chi_h(S_1)\equiv 2^{\rho-1} ({\rm mod }\ 2^\rho)$ for any $h\in H$ by (\ref{f-3.5}). Thus $\sum_{h\in H}\chi_h(S_1)\equiv N\cdot 2^{\rho-1} ({\rm mod }\ 2^\rho)\equiv 2^{\rho-1} ({\rm mod }\ 2^\rho)$ since $2\nmid N$. But
\begin{equation*}
  \sum_{h\in H}\chi_h(S_1)=\sum_{z\in S_1}\sum_{h\in H}\chi_h(z)=0 (\mbox{ since $a\not\in S_1$})
\end{equation*}
We get a contradiction $0\equiv 2^{\rho-1} ({\rm mod }\ 2^\rho)$.

If $a\in S$, then $2|d_0,2\nmid d_1$. By $v_2(d_1)=\rho-1$, we get $\rho=1$. Then by (\ref{f-3.5}), $\chi_h(S_1)$ is odd for all $h\in H$. Consider the following mapping
\begin{equation*}
  f: H\rightarrow \mathbb{C}, f(z)=\left\{\begin{array}{ll}
                                            1 & \mbox{ if $(1,z)\in S_1$} \\
                                            0 & \mbox{ otherwise}.
                                          \end{array}
  \right.
\end{equation*}
For $(1,z)\in G_1$, $h\in H$, $\chi_h((1,z))=\chi_h(z)$. The Fourier transformation of $f$ over group $H$ is
\begin{equation*}
  F(h)=\sum_{z\in H}f(z)\chi_h(z)=\sum_{z\in S_1}\chi_h(z)=\chi_h(S_1)\equiv 1 ({\rm mod }\ 2) \mbox{ for all $h\in H$}.
\end{equation*}
By inverse Fourier transformation,
\begin{equation*}
  Nf(z)=\sum_{h\in H}F(h)\overline{\chi_h(z)}\equiv \sum_{h\in H}\overline{\chi_h(z)} ({\rm mod }\ 2).
\end{equation*}
If $z\neq 0$, then $\sum_{h\in H}\overline{\chi_h(z)}=0$ and $f(z)\equiv 0 ({\rm mod }\ 2)$ and thus $f(z)=0$. Thus $(1,z)\not\in S_1$ for any $0\neq z\in H$. This means that $S_1=\{a=(1,0)\}$. Similarly, from (\ref{f-3.5}), we know that $v_2(\alpha_x-\alpha_{x'})\geq \rho+1=2$ for any $x,x'\in G_0$. Since $\alpha_0=d_0+d_1$ is odd, we know that $v_2(\alpha_x)=0$ for each $x\in G_0$. Then by $\alpha_x=\chi_h(S_0)+\chi_h(S_1)$ where $x=(0,h), h\in H$, and $2\nmid \chi_h(S_1)$, we get $2|\chi_h(S_0)$ for all $h\in H$. Consider the following mapping
\begin{equation*}
   f': H\rightarrow \mathbb{C}, f'(z)=\left\{\begin{array}{ll}
                                            1 & \mbox{ if $(0,z)\in S_0$} \\
                                            0 & \mbox{ otherwise}.
                                          \end{array}
  \right.
\end{equation*}
The Fourier transformation of $f$ over group $H$ is
\begin{equation*}
  F'(h)=\sum_{z\in H}f'(z)\chi_h(z)=\sum_{z\in S_0}\chi_h(z)=\chi_h(S_0)\equiv 0 ({\rm mod }\ 2) \mbox{ for all $h\in H$}.
\end{equation*}
By inverse Fourier transformation,
\begin{equation*}
  Nf'(z)=\sum_{h\in H}F'(h)\overline{\chi_h(z)}\equiv 0 ({\rm mod }\ 2).
\end{equation*}
Thus $f'(z)$ is even for all $z\in H$. namely, $S_0$ is empty and $S=\{a\}$. Since $a=(1,0)$, $|\langle S\rangle|=2<|G|$, thus $\Gamma$ is not connected. This completes the proof.
\end{proof}

\section{Cubelike graphs}

Let $G$ be the additive group of $\mathbb{F}_2^n$. For any subset $S$ of $\mathbb{F}_2^n$, $0\not\in S$, $\Gamma={\rm Cay}(G,S)$ is an integral simple graph and the order of any non-zero element in $G$ is two. The character group of $G$ is
\begin{equation*}
  \hat{G}=\hat{\mathbb{F}_2^n}=\{\chi_z: z\in \mathbb{F}_2^n\},
\end{equation*}
where for $g=(g_1,\cdots,g_n), z=(z_1,\cdots,z_n)\in \mathbb{F}_2^n$,
\begin{equation*}
  \chi_z(g)=(-1)^{z\cdot g}, z\cdot g=\sum_{j=1}^nz_jg_j\in \mathbb{F}_2.
\end{equation*}
If we view $\mathbb{F}_2^n$ as the additive group of the finite field $\mathbb{F}_{q}$ with $q=2^n$, then
\begin{equation*}
  \hat{G}=\widehat{(\mathbb{F}_q, +)}=\{\lambda_z:z\in \mathbb{F}_q\}
\end{equation*}
where for $g,z\in \mathbb{F}_q$, $\lambda_z(g)=(-1)^{T(zg)}$, and $T: \mathbb{F}_q\rightarrow \mathbb{F}_2$ is the trace mapping.

\begin{lem}\label{lem-4.1} Let $\Gamma={\rm Cay}(\mathbb{F}_q,S)$ be a connected graph. $S\subseteq \mathbb{F}_q^*=\mathbb{F}_q\setminus \{0\}$, $q=2^n\geq 2$, $a\in \mathbb{F}_q^*$. Then for each $c\in \mathbb{F}_q^*$, $\Gamma$ has PST between vertices $g$ and $g+a$ at time $t$ if and only if $\Gamma'={\rm Cay}(\mathbb{F}_q,S')$ has PST between $g'$ and $g'+a'$ at time $t$ where
\begin{equation*}
  S'=cS=\{cz:z\in S\}, g'=c^{-1}g, a'=c^{-1}a.
\end{equation*}
\end{lem}
\begin{proof}
For $z\in \mathbb{F}_q$, let $d=|S|=|S'|$ and
\begin{equation*}
  \alpha_z=\lambda_z(S)=\sum_{u\in S}\lambda_z(u)=\sum_{u\in S}(-1)^{T(zu)}, \alpha'_z=\lambda_z(S')=\sum_{u\in S'}(-1)^{T(uz)}.
\end{equation*}
Denote
\begin{eqnarray*}
  G_0=\{x\in \mathbb{F}_q: \lambda_z(a)=1\}, & & G_1=\{x\in \mathbb{F}_q: \lambda_z(a)=-1\}, \\
 G'_0=\{x\in \mathbb{F}_q: \lambda_z(a)=1\}, & & G'_1=\{x\in \mathbb{F}_q: \lambda'_z(a)=-1\},
\end{eqnarray*}
and
\begin{equation*}
  M_\varepsilon=\gcd(d-\alpha_z: z\in G_\varepsilon), M'_\varepsilon=\gcd(d-\alpha'_z: z\in G'_\varepsilon), (\varepsilon=0,1).
\end{equation*}
It is easy to see that $\alpha'_z=\alpha_{c^{-1}z}, G'_\varepsilon=c^{-1}G_\varepsilon$ ($\varepsilon=0,1$), and then $M_\varepsilon=M'_\varepsilon$ ($\varepsilon=0,1$). By Theorem \ref{thm-main},
\begin{eqnarray*}
   && \mbox{$\Gamma$ has PST between $g$ and $g+a$ at time $t$}\\
  \Leftrightarrow &&\mbox{ there exists $\rho\geq 0$ such that}\\
  &&\mbox{ $v_2(d-\alpha_y)=\rho$ for all $y\in G_1$, $v_2(M_0)\geq \rho+1$ and $t\in \frac{\pi}{M}+\frac{2\pi}{M}\mathbb{Z}$}\\
  && \mbox{ where $M=\gcd(M_0,M_1)$}\\
  \Leftrightarrow&& \mbox{there exists $\rho\geq 0$ such that}\\
  &&\mbox{ $v_2(d-\alpha'_y)=\rho$ for all $y\in G_1$, $v_2(M_0')\geq \rho+1$ and $t\in \frac{\pi}{M}+\frac{2\pi}{M}\mathbb{Z}$}\\
 \Leftrightarrow &&\mbox{$\Gamma'$ has PST between $g'$ and $g'+a'$ at time $t$}.
\end{eqnarray*}
\end{proof}

By Lemma \ref{lem-4.1}, we can consider a fixed non-zero element $a$ of $G$. From now on we consider $G=\mathbb{F}_2^n$ and $a=(1,0,\cdots,0)\in G$. It is proved in \cite{7} that if $\sum_{z\in S}z=a$, then $\Gamma={\rm Cay}(G,S)$ has PST between $g$ and $g+a$ for any $g\in G$ at time $\frac{\pi}{2}$. Cheung and Godsil \cite{10} present characterizations of PST in cubelike graphs in terms of binary codes. As an application they present a cubelike graph over $\mathbb{F}_2^5$ having PST at time $\frac{\pi}{4}$. They asked if the minimum time can be less than $\frac{\pi}{4}$. Godsil \cite{16} raised a question: Are there cubelike graphs having PST at time $\tau$ where $\tau$ is arbitrarily small? In the remain part of this section we will give positive answers to above questions. By Theorem \ref{thm-main1} and \ref{thm-main}, the minimum time of PST in cubelike graph $\Gamma={\rm Cay}(G,S)$ is $\frac{\pi}{M}$ where $M=\gcd(d-\alpha_z:z\in G)$, $d=|S|$ and $\alpha_z=\chi_z(S)$. Firstly, we show that $M$ should be a power of $2$.

\begin{lem}\label{lem-4.2} Let $G=\mathbb{F}_2^n$, $0\not\in S\subseteq \mathbb{F}_2^n$, $d=|S|\geq 1$. Then $M=\gcd(d-\alpha_z: z\in G)$ is a power of $2$ where $\alpha_z=\chi_z(S)$.\end{lem}

\begin{proof} Suppose that $p$ is an odd prime and $p|M$. Then $d-\alpha_z\equiv 0 ({\rm mod }\ p)$ for all $z\in G$. Thus
\begin{equation*}
  0\equiv \sum_{z\in G}(d-\alpha_z)=d\cdot 2^n-\sum_{z\in G}\alpha_z ({\rm mod }\ p).
\end{equation*}
But
\begin{equation*}
  \sum_{z\in G}\alpha_z=\sum_{u\in S}\sum_{z\in G}\chi_z(u)=\sum_{u\in S}0=0\ \ (\mbox{ since $0\not\in S$}).
\end{equation*}
We get that $d\cdot 2^n\equiv 0  ({\rm mod }\ p)$ and then $d\equiv 0  ({\rm mod }\ p)$. Therefore $\alpha_z\equiv0  ({\rm mod }\ p)$ for all $z\in G$.
Consider the mapping $f: \mathbb{F}_2^n\rightarrow \mathbb{Z}$ defined by
\begin{equation*}
  f(z)=\left\{\begin{array}{ll}
                1& \mbox{ if $g\in S$} \\
                0 & \mbox{ otherwise.}
              \end{array}
  \right.
\end{equation*}
The Fourier transformation of $f$ is
\begin{equation*}
  F(z)=\sum_{g\in G}f(g)\chi_z(g)=\sum_{g\in S}\chi_z(g)=\alpha_z\equiv 0  ({\rm mod }\ p).
\end{equation*}
Then by inverse Fourier transformation,
\begin{equation*}
  2^nf(g)=\sum_{z\in G}F(z)\chi_z(g)\equiv 0  ({\rm mod }\ p).
\end{equation*}
From $d=|S|\geq 1$, we know that there exists $g\in S$ and then $f(g)=1$. We get a contradiction $2^n\equiv 0  ({\rm mod }\ p)$. This completes the proof of Lemma \ref{lem-4.2}.
\end{proof}

\begin{thm}\label{thm-4.3}Let $G=\mathbb{F}_2^n$, $n\geq 2$, $0\not\in S\subset G$ and $\Gamma={\rm Cay}(G,S)$ be a connected graph, $0\neq a\in G$. If $\Gamma$ has PST between $g$ and $g+a$ at time $t$, then the minimum time $t$ is $\frac{\pi}{M}$, where $M=2^\ell$, $1\leq \ell \leq [\frac{n}{2}]$, where $[\frac{n}{2}]=s$ if $n=2s+1$ or $2s$.\end{thm}

\begin{proof}Let $a=(1,0), 0\in \mathbb{F}_2^{n-1}$ without loss of generality by Lemma \ref{lem-4.1}. Each element of $\mathbb{F}_2^n$ can be written as $z=(z_1,z')$, $z'\in \mathbb{Z}_2^{n-1}$. Then
\begin{eqnarray*}
  G_0 &=& \{z\in \mathbb{F}_2^n: \chi_z(a)=1\}=\{(z_1,z'): z_1\in \mathbb{F}_2, z'\in \mathbb{F}_2^{n-1}, 1=(-1)^{z_1}\}=(0,\mathbb{F}_2^{n-1}) \\
  G_1 &=&\{z\in \mathbb{F}_2^n: \chi_z(a)=-1\}=(1,\mathbb{F}_2^{n-1}).
\end{eqnarray*}
Suppose that $\Gamma$ has PST between $g$ and $g+a$. By Theorem \ref{thm-main}, there exists $\rho\geq 0$ such that $v_2(d-\alpha_y)=\rho$ for all $y\in G_1$ and $v_2(d-\alpha_x)\geq \rho+1$ for all $x\in G_0$, where $d=|S|$. Therefore $\min\{v_2(d-\alpha_y),y\in G\}=\rho$ and by Lemma \ref{lem-4.2}, $2^\ell=M=2^\rho$. Namely, we get $\ell=\rho$.

Below we prove that $1\leq \ell \leq [\frac{n}{2}]$.

Let $S_\varepsilon=S\cap G_\varepsilon$ ($\varepsilon=0,1$). Then $S_\varepsilon=(\varepsilon,S'_\varepsilon)$ where $S'_\varepsilon=\{z'\in \mathbb{F}_2^{n-1}: (\varepsilon, z')\in S\}$. Let $d_\varepsilon=|S_\varepsilon|=|S'_\varepsilon|$. Then $S=S_0\cup S_1$, $d=d_0+d_1$, $d_\varepsilon=|S_\varepsilon|$. For $x=(0,x')\in G_0$,
\begin{equation*}
  \alpha_x=\sum_{g'\in S'_0}\chi_{x'}(g')+\sum_{g'\in S'_1}\chi_{x'}(g')=\chi_{x'}(S'_0)+\chi_{x'}(S'_1).
\end{equation*}
Similarly, for $y=(1,y')\in G_1$, $\alpha_y=\chi_{y'}(S'_0)-\chi_{y'}(S'_1)$. From Corollary \ref{cor-2.5} we have
\begin{equation*}
  v_2(\alpha_x-\alpha_y)=\rho \mbox{ for all $x\in G_0$, $y\in G_1$}.
\end{equation*}
Particularly, taking $x=(0,z'), y=(1,z')$, we get
\begin{equation*}
  \rho=v_2((\chi_{z'}(S'_0)+\chi_{z'}(S'_1))-(\chi_{z'}(S'_0)-\chi_{z'}(S'_1)))=v_2(2\chi_{z'}(S'_1)).
\end{equation*}
Namely, $v_2(\chi_{z'}(S'_1))=\rho-1$ for all $z'\in \mathbb{F}_2^{n-1}$. Thus $\rho\geq 1$ and
\begin{equation}\label{f-4.1}
  \chi_{z'}(S'_1)=2^{\rho-1}\theta(z'), 2\nmid \theta(z')\in \mathbb{Z}.
\end{equation}
Moreover, we have $0\leq d_1=|S'_1|\leq |\mathbb{F}_2^{n-1}|=2^{n-1}$. If $d_1=2^{n-1}$, then $S'_1=\mathbb{F}_2^{n-1}$ and for $0\neq z'\in \mathbb{F}_2^{n-1}$ (such $z'$ exists by assumption $n\geq 2$), $\chi_{z'}(S'_1)=\chi_{z'}(\mathbb{F}_2^{n-1})=0$ which is contrary to $v_2(\chi_{z'}(S'_1))=\rho-1$. Therefore, we get $0\leq d_1\leq 2^{n-1}-1$ and then $v_2(d_1)\leq n-2$. Now we consider
\begin{eqnarray*}
  \sum_{z'\in \mathbb{F}_2^{n-1}}\chi_{z'}(S'_1)^2 &=& \sum_{g,h\in S'_1}\sum_{z'\in \mathbb{F}_2^{n-1}}\chi_{z'}(g-h) \\
 &=& d_1^2+\sum_{g,h\in S'_1}\sum_{z'\in \mathbb{F}_2^{n-1},z'\neq0}\chi_{z'}(g-h)\\
 &=&d_1^2+\sum_{g\in S'_1}\sum_{z'\in \mathbb{F}_2^{n-1},z'\neq0}\chi_{z'}(g-g)+\sum_{g\neq h\in S'_1}\sum_{z'\in \mathbb{F}_2^{n-1},z'\neq0}\chi_{z'}(g-h)\\
 &=&d_1^2+d_1\cdot (2^{n-1}-1)+(-1)d_1(d_1-1)\\
 &=&d_1\cdot 2^{n-1}.
\end{eqnarray*}
On the other hand, by (\ref{f-4.1}), we have
\begin{equation*}
  \sum_{z'\in \mathbb{F}_2^{n-1}}\chi_{z'}(S'_1)^2=2^{2\rho-2}\sum_{z'\in \mathbb{F}_2^{n-1}}\theta(z')^2\geq 2^{2\rho-2}2^{n-1} \ (\mbox{ since $2\nmid \theta(z')$}).
\end{equation*}
Thus $2^{2\rho-2}\leq d_1$ and $2\rho-2\leq v_2(d_1)\leq n-2$. Therefore $\ell =\rho\leq \frac{n}{2}$. This completes the proof of Theorem \ref{thm-4.3}.
\end{proof}

Now we show that for $n=2m+1$ ($m\geq 2$), the minimum time $\frac{\pi}{M}$ can reach the lower bound $\frac{\pi}{2^m}$ given by Theorem \ref{thm-4.3}. For doing this we use the one to one corresponding between subsets $S$ of $\mathbb{F}_2^n$ and Boolean functions $f: \mathbb{F}_2^n\rightarrow \mathbb{F}_2$ with $n$ variables by
\begin{equation*}
  f(x)=\left\{\begin{array}{cc}
                1 & \mbox{ if $x\in S$} \\
               0 & \mbox{ otherwise}.
              \end{array}
  \right.
\end{equation*}
The set $S$ is called the support of $F$ and denoted by $S={\rm supp}(f)$.

Let $B_n$ be the ring of Boolean functions with $n$ variables. For $f\in B_n$, we have the function $(-1)^f: \mathbb{F}_2^n\rightarrow \{\pm 1\}\subset \mathbb{Z}$ where $x\mapsto (-1)^{f(x)}$. The Fourier transformation of $(-1)^f$ over the group $(\mathbb{F}_2^n,+)$, also called the Walsh transformation of $f$, is $W_f: \mathbb{F}_2^n\rightarrow \mathbb{Z}$ defined by
\begin{equation*}
  W_f(y)=\sum_{x\in \mathbb{F}_2^n}(-1)^{f(x)+x\cdot y}\ \ (y\in \mathbb{F}_2^n).
\end{equation*}
\begin{defn}\label{defn-4.4}For $n=2m$ ($m\geq 2$), $f\in B_n$ is called a bent function if $|W_f(y)|=2^m$ for all $y\in \mathbb{F}_2^n$.\end{defn}

Bent function exists for all $m\geq 1$ and many series of bent functions have been constructed in past forty years.

\begin{thm}\label{thm-4.5}Let $n=2m+1$ ($m\geq 2$), $f$ be a bent function in $B_{n-1}=B_{2m}$, $S'={\rm supp}(f)=\{z'\in \mathbb{F}_2^{n-1}: f(z')=1\}$, $0\not\in S'$ (instead of $f$ by $f+1$ if necessary), $S_\varepsilon=(\varepsilon,S')$ ($\varepsilon=0,1$) and $S=S_0\cup S_1$. Then

(1) The cubelike graph $\Gamma={\rm Cay}(\mathbb{F}_2^n,S)$ is connected.

(2) For $a=(1,0_{n-1})$, $0_{n-1}\in \mathbb{F}_2^{n-1}$, $\Gamma$ has PST between $g$ and $g+a$ for any $g\in \mathbb{F}_2^n$ at time $\frac{\pi}{2^m}$.

(3) The minimum period of any vertex in $\Gamma$ is $\frac{\pi}{2^m}$.\end{thm}

\begin{proof}From $0_{n-1}\not\in S'$ we know that $0=(0,0_{n-1})\not\in S$ so that $\Gamma$ is a simple graph. It is well-known that $S'$ is a difference set of $(\mathbb{F}_2^{n-1},+)$ with parameters $(v,k,\lambda)$ where $v=2^{n-1}, k=2^{n-2}\pm 2^{m-1}=|S'|$, $\lambda=2^{n-3}\pm 2^{m-1}$. Each non-zero element of $\mathbb{F}_2^{n-1}$ occurs exactly $\lambda$ times in the multiset $\{g-g': g,g'\in S'\}$. By assumption $m\geq 2$, $\lambda\geq 2^{m-3}-2^{m-1}\geq 1$. Thus $\langle S'\rangle=\mathbb{F}_2^{n-1}$ and then $\langle S_0\rangle=\langle (0,S')\rangle=(0, \mathbb{F}_2^{n-1})$. Since $|S_1|=|S'|=k\geq 1$, there exists element $(1,z')$ in $S_1$. Then $\langle S\rangle=\langle S_0\cup S_1\rangle=\mathbb{F}_2^n$ which means that the graph $\Gamma$ is connected.

For $x=(0,x')\in G_0=(0,\mathbb{F}_2^{n-1})$, $y=(1,y')\in G_1=(1, \mathbb{F}_2^{n-1})$,
\begin{eqnarray*}
  \alpha_x &=& \sum_{z\in S_0}\chi_x(z)+\sum_{z\in S_1}\chi_x(z)=\sum_{z'\in S'}\chi_{x'}(z')+\sum_{z'\in S'}\chi_{x'}(z') \\
   &=& 2\chi_{x'}(S') \\
   \alpha_y&=&\sum_{z\in S_0}\chi_y(z)+\sum_{z\in S_1}\chi_y(z)=\sum_{z'\in S'}\chi_{y'}(z')-\sum_{z'\in S'}\chi_{y'}(z')=0.
\end{eqnarray*}
Therefore,
\begin{equation}\label{f-4.2}
 v_2(\alpha_x-\alpha_y)=v_2(2\chi_{x'}(S')) \mbox{ for each $x'\in \mathbb{F}_2^{n-1}$}.
\end{equation}
But
\begin{eqnarray}
 \nonumber 2\chi_{x'}(S') &=& 2\sum_{z\in S'}(-1)^{zx'}=\sum_{z\in \mathbb{F}_2^{n-1}}(1-(-1)^{f(z)})(-1)^{zx'}=\sum_{z\in \mathbb{F}_2^{n-1}}(-1)^{zx'}-W_f(x') \\
 \label{f-4.3} &=& \left\{\begin{array}{ll}
               2^{n-1}-W_f(x') & \mbox{ if $x'=0\in \mathbb{F}_2^{n-1}$} \\
               -W_f(x') & \mbox{ otherwise.}
             \end{array}
  \right.
\end{eqnarray}
Since $f$ is a bent function in $B_{n-1}$, we know that $W_f(x')=2^m$ or $-2^m$. Then by $n-1=2m>m$ and formula (\ref{f-4.3}) we know that $v_2(2\chi_{x'}(S'))=m$ for all $x'\in \mathbb{F}_2^{n-1}$. By (\ref{f-4.2}) we get $v_2(\alpha_x-\alpha_y)=m$ for all $x\in G_0$ and $y\in G_1$ which means that $\Gamma$ has PST between $g$ and $g+a$ at time $\frac{\pi}{2^m}$.

At last, since $v_2(d-\alpha_y)=v_2(\alpha_0-\alpha_y)=m$ for any $y\in G_1$ and $v_2(d-\alpha_x)\geq m+1$ for any $x\in G_0$, we know that the minimum period of any vertex is $\frac{\pi}{2^m}$ by Lemma \ref{lem-4.2}. This completes the proof of Theorem \ref{thm-4.5}.\end{proof}
\begin{thm}\label{thm-4.6}Let $n=2m$ $m\geq 2$, $f(x)$ be a bent function with $n$ variables, $f(0)=0$, $S={\rm supp}(f)=\{x\in \mathbb{F}_2^n: f(z)=1\}$. Then the cubelike graph $\Gamma={\rm Cay}(\mathbb{F}_2^n,S)$ is connected and the minimum period of each vertex $g\in \mathbb{F}_2^n$ in $\Gamma$ is $\frac{\pi}{2^m}$.\end{thm}

\begin{proof}By assumption $f(0)=0$ we know that $0\not\in S$ so that $\Gamma$ is a simple graph. $S$ is a difference set in $(\mathbb{F}_2^n,+)$ with parameters $(v,k,\lambda)$ where $v=2^n$, $k=|S|=2^{n-1}\pm 2^{m-1}$ and $\lambda=2^{n-2}\pm 2^{m-1}$. From $m\geq 2$, we know that each non-zero element in $\mathbb{F}_2^n$ can be expressed as $g-g'$ ($g,g'\in S$) in $\lambda\geq 2^{2m-2}-2^{m-1}\geq 1$ ways, thus $\langle S\rangle=\mathbb{F}_2^n$ and $\Gamma$ is a connected graph. Moreover, as shown in the proof of Theorem \ref{thm-4.5}, for each $z\in \mathbb{F}_2^n$
\begin{equation*}
  2\alpha_z=2\sum_{g\in S}(-1)^{zg}=\left\{\begin{array}{cc}
                                             2^n-W_f(z) & \mbox{ if $z=0$} \\
                                             -W_f(z) & \mbox{ otherwise},
                                           \end{array}
  \right.
\end{equation*}
where $W_f(z)$ is the Walsh transformation of $F$. By assumption, $f$ is bent, we know that $W_f(z)=2^m$ or $-2^m$. Since $d=|S|=k=2^{n-1}\pm 2^{m-1}=\alpha_0$, we get $d-\alpha_0=0$ and for $z\in \mathbb{F}_2^n, z\neq 0$
\begin{equation*}
  2(d-\alpha_z)=2^n\pm 2^m-2\alpha_z=2^n \mbox{ or } 2^n-2^{m+1},
\end{equation*}
and there exists $z$ such that $2(d-\alpha_z)=2^n+2^{m+1}$ or $2^n-2^{m+1}$. Since $n=2m\geq m+1$, we know that $M=\gcd(d-\alpha_z: z\in \mathbb{F}_2^n)$ is $2^m$ by Lemma \ref{lem-4.2}. Therefore the minimum period of each vertex $g\in \mathbb{F}_2^n$ in $\Gamma$ is $\frac{\pi}{M}=\frac{\pi}{2^m}$. \end{proof}

\section{Conclusion}
In this paper we present a characterization of connected abelian Cayley graphs having perfect state trnafer (PST) between two distinct vertices (Theorem \ref{thm-main}). As applications of this characterization, we prove the following results.

(I) If an integral connected abelian graph $\Gamma={\rm Cay}(G,S)$ has PST between distinct vertices, then $4||G|$ (Theorem \ref{thm-3.5}). Such result has been proved for circulant graphs in \cite{25}.

(II) Let $\Gamma={\rm Cay}(\mathbb{F}_2^n, S)$ be a connected cubelike graph and $n\geq 2$. If $\Gamma$ has PST between two vertices $g$ and $h$ at minimum time $t$, then $t=\frac{\pi}{2^\ell}$, $\ell\geq 1$ (Lemma \ref{lem-4.2}). Moreover, if $g\neq h$, then $\ell\leq [\frac{n}{2}]=m$ for $n=2m$ or $2m+1$.

(III) For $n\geq 4$, there exists a connected cubelike graph $\Gamma={\rm Cay}(\mathbb{F}_2^n,S)$ such that each vertex in $\Gamma$ has minimum period $\frac{\pi}{2^m}$ (Thereom \ref{thm-4.5} for $n=2m+1$ and Theorem \ref{thm-4.6} for $n=2m$).

(IV) For $n=2m+1$ ($m\geq 2$), $0\neq a\in \mathbb{F}_2^n$, there exists a connected cubelike graph $\Gamma={\rm Cay}(\mathbb{F}_2^n, S)$ such that for each $g\in \mathbb{F}_2^n$, $\Gamma$ has PST between $g$ and $g+a$ with minimum time $\frac{\pi}{2^m}$ (Theorem \ref{thm-4.5}).

The results (III) and (IV) are proved by using bent functions. In general speaking, if we know the distribution of Walsh values of a Boolean function $f\in B_n$ with $f(0)=0$, we can determine the minimum period of each vertex of the cubelike graph $\Gamma={\rm Cay}(\mathbb{F}_2^n,S)$ for $S={\rm supp}(f)$, and the minimum time of PST between two distinct vertices in $\Gamma$ if such PST occurs. On the other hand, for $n=2m$ $(m\geq 4)$ we do not know if the minimum time of PST between two distinct vertices of $\Gamma={\rm Cay}(\mathbb{F}_2^n,S)$ can reach the lower bound $\frac{\pi}{2^m}$. We do not know either that if the minimum period of vertices in $\Gamma={\rm Cay}(\mathbb{F}_2^n,S)$ can be less than $\frac{\pi}{2^m}$, $m=[\frac{n}{2}]$.

\end{document}